\documentclass[11pt]{article}
\usepackage{libertine, amsmath, amsthm, amsfonts, verbatim, mathabx, breqn, mathtools, color, graphicx, tikz}
\usepackage[pdftex,colorlinks=true]{hyperref}
\usepackage[utf8]{inputenc}

\newcommand{\qlb}{\mathrm{QLB}}
\newcommand{\lb}{\mathrm{LB}}
\newcommand{\itlb}{\mathrm{ITLB}}
\newcommand{\QH}{\mathrm{QH}}
\renewcommand{\H}{\mathrm{H}}
\newcommand{\sort}{\textsc{Sort}}
\newcommand{\C}{{\cal C}}

\newcommand{\ket} [1] {\vert #1 \rangle}
\newcommand{\bra} [1] {\langle #1 \vert}
\newcommand{\braket}[2]{\langle #1 | #2 \rangle}

\newcommand{\ketbra}[2]{\ket{#1}\bra{#2}}

\newcommand{\Exp}[2]{\mathbf{E}_{#1}\left[#2\right]}
\newcommand{\norm}[1]{\left\lVert#1\right\rVert}

\newcommand{\Prob}[2]{\mathrm{Prob}_{#1}\left[#2\right]}

\newcommand{\adv}{\mathrm{Adv}}

\newcommand{\abs}[1]{\mid\! #1\! \mid}

\definecolor{darkred}{rgb}{.4,0,0}
\renewcommand{\em}{\color{darkred}}

\newtheorem{theorem}{Theorem}
\newtheorem{lemma}{Lemma}

\newtheorem{corollary}{Corollary}

\usetikzlibrary{arrows,shapes,graphs}
\tikzstyle{vertex}=[circle,draw=black,fill=orange!40,minimum size=12pt,inner sep=1pt]
\tikzstyle{point}=[circle,draw=black,fill=blue,minimum size=3pt,inner sep=1pt]
\tikzstyle{edge} = [draw,thick,-]
\tikzstyle{dashed edge} = [draw,thick,dashed,-]
\tikzstyle{flip} = [draw,line width=3pt,->]
\tikzstyle{arrow} = [draw,thick,->]
\tikzstyle{back} = [dotted]
\tikzstyle{facet} = [fill=yellow,fill opacity=0.2]
\tikzstyle{bluefacet} = [fill=blue,fill opacity=0.4]

\title{Information-theoretic lower bounds\\
  for quantum sorting}
\author{Jean Cardinal \and Gwena\"el Joret \and J\'er\'emie Roland}
\date{Universit\'e libre de Bruxelles (ULB)\\
\today}
\begin{document}
\maketitle
\sloppy
\begin{abstract}
We analyze the quantum query complexity of sorting under partial information.
In this problem, we are given a partially ordered set $P$ and are asked to identify a linear extension
of $P$ using pairwise comparisons.
For the standard sorting problem, in which $P$ is empty, it is known that the quantum query complexity is not asymptotically
smaller than the classical information-theoretic lower bound.
We prove that this holds for a wide class of partially ordered sets, thereby improving on a result from Yao (STOC'04).
\end{abstract}

\section{Introduction}

Sorting by comparison is a well-studied computational problem in which a permutation $\sigma$ of $n$ elements is to be identified by asking
questions of the form ``is $\sigma (i)\leq \sigma (j)$?''. The complexity of a sorting algorithm is the number of such comparisons it performs in the worst case as a function of $n$. Optimal algorithms are known, solving the problem using $O (n\log n)$ comparisons.
We consider a generalization of the sorting problem, called {\em sorting under partial information}, in which we are given a partially ordered set $P$, and the goal is to identify a permutation $\sigma$ such that $i\leq_P j\implies \sigma(i)\leq \sigma(j)$.
Such a permutation is called a {\em linear extension} of $P$.
Here $P$ is a given partial information on the sought permutation, and can be thought of as a set of comparisons whose outcomes are already known.
An illustration is given in Figure~\ref{fig:example}.
We denote this problem by $\sort_P$. It generalizes many standard comparison-based problems such as insertion in a sorted list, merging sorted lists, and sorting elements from a static data structure such as a heap.

\begin{figure}
\begin{center}
\begin{tikzpicture}[scale=.7,auto,swap]
\path[flip] (3.5,2.5) -- (4.2,2.5);
\path[flip] (7.5,2.5) -- (8.2,2.5);
\path[flip] (11.5,2.5) -- (12.2,2.5);
\path[flip] (15.5,2.5) -- (16.2,2.5);
\begin{scope}[xshift=0cm]
\foreach \pos/\name/\outdeg in {{(1,4)/a/a}, {(3,4)/b/b}, {(3,2)/c/c}, {(1,3)/d/d}, {(1,2)/e/e}, {(2,1)/f/f}}
   \node[vertex] (\name) at \pos {$\outdeg$};
\foreach \source/\dest in {a/d,d/e,e/f,c/f,b/c,b/e}
   \path[edge] (\source) -- (\dest);
\path[dashed edge] (a) -- (b);
\end{scope}
\begin{scope}[xshift=4cm]
\foreach \pos/\name/\outdeg in {{(2,4)/a/a}, {(3,3)/b/b}, {(3,2)/c/c}, {(1,3)/d/d}, {(1,2)/e/e}, {(2,1)/f/f}}
   \node[vertex] (\name) at \pos {$\outdeg$};
\foreach \source/\dest in {a/d,a/b,d/e,e/f,c/f,b/c,b/e}
   \path[edge] (\source) -- (\dest);
\path[dashed edge] (d) -- (b);
\end{scope}
\begin{scope}[xshift=8cm]
\foreach \pos/\name/\outdeg in {{(2,5)/a/a}, {(2,4)/b/b}, {(3,2)/c/c}, {(1,3)/d/d}, {(1,2)/e/e}, {(2,1)/f/f}}
   \node[vertex] (\name) at \pos {$\outdeg$};
\foreach \source/\dest in {a/b,b/c,b/d,d/e,e/f,c/f}
   \path[edge] (\source) -- (\dest);
\path[dashed edge] (e) -- (c);
\end{scope}
\begin{scope}[xshift=12cm]
\foreach \pos/\name/\outdeg in {{(2,5)/a/a}, {(2,4)/b/b}, {(3,3)/c/c}, {(1,3)/d/d}, {(1,2)/e/e}, {(2,1)/f/f}}
   \node[vertex] (\name) at \pos {$\outdeg$};
\foreach \source/\dest in {a/b,b/c,b/d,d/e,e/f,c/e}
   \path[edge] (\source) -- (\dest);
\path[dashed edge] (d) -- (c);
\end{scope}
\begin{scope}[xshift=16cm]
\foreach \pos/\name/\outdeg in {{(1,5)/a/a}, {(1,4)/b/b}, {(1,3)/c/c}, {(1,2)/d/d}, {(1,1)/e/e}, {(1,0)/f/f}}
   \node[vertex] (\name) at \pos {$\outdeg$};
\foreach \source/\dest in {a/b,b/c,c/d,d/e,e/f}
   \path[edge] (\source) -- (\dest);
\end{scope}
\end{tikzpicture}
\end{center}
\caption{\label{fig:example}An instance of the problem of sorting under partial information. Here four successive comparisons
suffice to unveil the underlying total order.}
\end{figure}
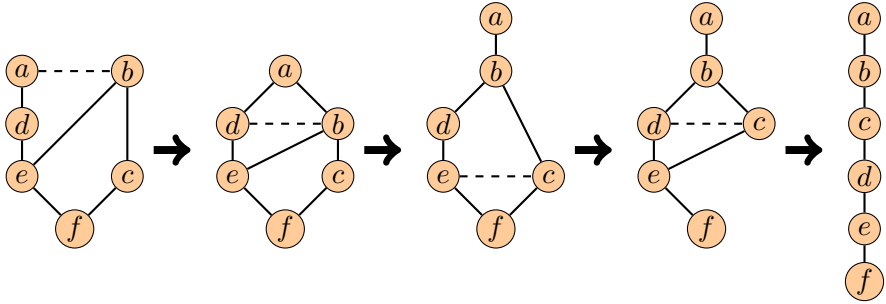

The problem has a long history, dating back to a seminal 1976 paper from Michael Fredman~\cite{F76}, 
and has found practical applications~\cite{DM18,HKSL19}.
There exist optimal algorithms performing $O(\log (|\Delta (P)|))$ comparisons, where
$\Delta(P)$ is the set of linear extensions of $P$. This is a consequence of the existence of so-called balanced pairs
in partial orders~\cite{KL91,BFT95,B99}. It is also known that an optimal sequence of comparisons can be found
in polynomial time~\cite{KK95,CFJJM13}. 
The expression $\log (|\Delta (P)|)$ is often referred to as the {\em information-theoretic lower bound}. 

We study the quantum query complexity of $\sort_P$, the minimum number of comparisons performed by any quantum decision tree
solving $\sort_P$. We refer the reader to the survey of Buhrman and de Wolf~\cite{BW02} for the definition of quantum decision trees.

For the standard sorting problem, H\o{}yer, Neerbek, and Shi~\cite{HNS02} proved that the quantum query complexity is bounded
from below by a constant times the information-theoretic lower bound $\log_2 (n!)$. Hence when $P$ is empty, no asymptotic quantum speedup
is achievable for $\sort_P$.
They also showed that the information-theoretic lower bound holds for the ordered search problem, another special case of $\sort_P$ in which $P$ is composed of a chain and an isolated element. The analysis of the quantum lower bound for ordered search was further refined by Childs and Lee~\cite{CL08}.

At STOC'04~\cite{Y04}, Yao gave the following lower bound for $\sort_P$.
\begin{theorem}[Yao~\cite{Y04}]
\label{thm:yao}
The quantum query complexity of sorting under partial information, given a poset $P$ on $n$ elements, is at least
$c\log (|\Delta (P)|) - c'n,$
for some universal constants $c,c'>0$.
\end{theorem}

This lower bound is clearly useless when $\log (|\Delta (P)|) < c'n/c$.
Therefore, some asymptotic quantum speedup could still be achievable
in cases where the information-theoretic lower bound is $o(n)$.
Our main result rules out this possibility and improves on Theorem~\ref{thm:yao} for a wide class of posets.

A poset is said to be {\em series-parallel} when it can be obtained by a series or parallel composition of smaller posets.
The $\sort_P$ problem restricted to series-parallel posets includes multiway merging, insertions of multiple elements, and sorting heap-ordered data as special cases.
There are families of arbitrarily dense series-parallel posets $P$ on $n$ elements with $\log (|\Delta (P)|)=o(n)$.

\begin{theorem}
\label{thm:main}
The quantum query complexity of sorting under partial information, given a series-parallel poset $P$, is at least $c\log (|\Delta (P)|)$ for some universal constant $c>0$.
\end{theorem}

Theorem~\ref{thm:main} is proved by relating a quantum adversary lower bound to the partial order entropy, refining an elegant connection established by Yao.
He conjectured that the information-theoretic lower bound holds for any partial order $P$ (up to a constant factor). 
As a further step in this direction, we show that our result does not crucially relies on $P$ being series-parallel. We generalize Theorem~\ref{thm:main} to a wider class of posets obtained by series and parallel compositions of posets that are in some precise sense far from being series-parallel.

The next section presents several instrumental notions from partial order combinatorics and information theory.
In Section~\ref{sec:qsort}, we formulate a lower bound on the quantum query complexity of $\sort_P$
and discuss the relation between our findings and the developments of Yao.
Finally, Section~\ref{sec:sp} gives the proof of our main result.

\section{Sorting and partial order entropy}

Throughout this paper, we denote a poset by a pair $P=(A, \leq_P)$ composed of a ground set $A$ of $n$ elements and a
partial order $\leq_P$ on $A$, defined as a reflexive, antisymmetric, transitive binary relation. 
For convenience in indexing, we often let $A=[n]$, where $[n]=\{1,2,\ldots ,n\}$.
The set $\Delta (P)$ of linear extensions of $P$ is the set of permutations $\sigma$ of $A$ corresponding to total orders extending $\leq_P$,
hence such that $i\leq_P j\implies \sigma (i)\leq \sigma (j)$. Here $\sigma (i)$ denotes the rank of element $i\in A$ in the total order.

The information-theoretic lower bound for sorting under a partial information $P$ is the logarithm of the number
of linear extensions of $P$. For convenience, we use natural logarithms, and refer to this lower bound as:
$$
\itlb (P) \coloneqq \ln (|\Delta (P)|) . 
$$

\paragraph*{Partial order entropy.}

The notion of entropy of a partial order plays a central role in recent advances on sorting problems~\cite{CFJJM10,CFJJM13,CF13,FR16}. 
We introduce the necessary background.
Consider a poset $P=([n],\leq_P)$.
A {\em chain} in $P$ is a sequence $i_1, i_2,\ldots i_k$ of elements in $[n]$ such that $i_1\leq_P i_2\leq_P\ldots\leq_P i_k$.
The {\em chain polytope} $\C (P)$ of $P$ is the subset of $\mathbb{R}^n$ defined by the points $y$ such that:
  \begin{eqnarray*}
    y_i & \geq & 0 \ \ \ \forall i\in [n]\\
    \sum_{i\in C} y_i & \leq & 1\ \ \ \text{ for every chain } C \text{ of } P .
    \end{eqnarray*}
An example is given in Figure~\ref{fig:chain}.
The {\em entropy} of $P$ is defined as 
\begin{equation}
\label{eqn:entropy}
\H(P) \coloneqq \min_{z \in \C(P)} -\frac 1n \sum_{i\in [n]} \ln z_i.
\end{equation}

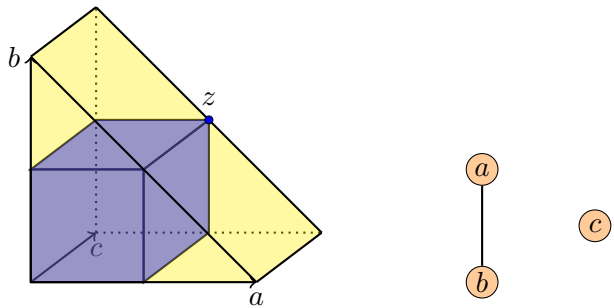
\begin{figure}
\begin{center}
  \begin{tikzpicture}[x={(1cm, 0cm)}, y={(0cm, 1cm)}, z={(0.289310cm, 0.219527cm)},scale=1.5]
\begin{scope}[xshift=0cm,scale=2]
\coordinate (1) at (0, 0, 0);
\coordinate (2) at (0, 0, 1);
\coordinate (3) at (0, 1, 0);
\coordinate (4) at (0, 1, 1);
\coordinate (5) at (1, 0, 0);
\coordinate (6) at (1, 0, 1);
\coordinate (1c) at (0, 0, 0);
\coordinate (2c) at (0, 0, 1);
\coordinate (3c) at (0, .5, 0);
\coordinate (4c) at (.5, 0, 0);
\coordinate (5c) at (.5, .5, 0);
\coordinate (6c) at (.5, 0, 1);
\coordinate (7c) at (0, .5, 1);
\coordinate (8c) at (.5, .5, 1);
\draw[edge,back] (2) -- (4);
\draw[edge,back] (2) -- (6);
\draw[edge] (1c) -- (4c) -- (5c) -- (3c) -- cycle {};
\draw[edge] (8c) -- (7c) -- (3c) -- (5c) -- cycle {};
\draw[edge] (4c) -- (6c) -- (8c) -- (5c) -- cycle {};
\draw[arrow] (1) -- (2) node[below] {$c$};
\fill[facet] (1) -- (3) -- (5) -- cycle {};
\fill[facet] (4) -- (2) -- (6) -- cycle {};
\fill[facet] (3) -- (4) -- (6) -- (5) -- cycle {};
\fill[facet] (3) -- (4) -- (2) -- (1) -- cycle {};
\fill[facet] (1) -- (2) -- (6) -- (5) -- cycle {};
\fill[bluefacet] (1c) -- (4c) -- (5c) -- (3c) -- cycle {};
\fill[bluefacet] (8c) -- (7c) -- (3c) -- (5c) -- cycle {};
\fill[bluefacet] (4c) -- (6c) -- (8c) -- (5c) -- cycle {};
\draw[edge] (3) -- (4);
\draw[edge] (5) -- (6);
\draw[edge] (4) -- (6);
\draw[edge] (3) -- (5);
\draw[arrow] (1) -- (5) node[below] {$a$};
\draw[arrow] (1) -- (3) node[left] {$b$};
\node[point,label=$z$] at (8c) {};
\end{scope}
\begin{scope}[xshift=3cm]
\foreach \pos/\name/\outdeg in {{(1,1)/a/a}, {(1,0)/b/b}, {(2,.5)/c/c}}
   \node[vertex] (\name) at \pos {$\outdeg$};
\draw[edge] (a) -- (b);
\end{scope}
\end{tikzpicture}
\end{center}
\caption{\label{fig:chain}The chain polytope of the poset $(\{a,b,c\},\{(b\leq a)\})$.
The point realizing the minimum in the definition of the entropy in Equation~\eqref{eqn:entropy} is $z=(1/2,1/2,1)$,
and the entropy is $\frac{2}{3}\ln 2\simeq 0.462$.}
\end{figure}

The underlying optimization problem consists of maximizing the volume $\prod_i z_i$ of an axis-aligned box
contained in $\C (P)$, having one of its corner at the origin and $z$ as opposite corner.
The entropy of a poset is a special case of {\em graph entropy}, where the graph is the comparability graph of $P$.
For further insights and applications of the notion of graph entropy,
the reader is referred to the survey of Simonyi~\cite{S93}.

Kahn and Kim~\cite{KK95} showed that the information-theoretic lower bound for $\sort_P$ is closely approximated by the following function of the entropy of $P$: 
\begin{equation}
\label{eqn:lb}
\lb (P) \coloneqq n (\ln n - \H(P)).
\end{equation}
Namely, they proved that $\itlb (P) \leq \lb (P) \leq c \cdot \itlb (P)$ for some constant $c>1$. 
Cardinal, Fiorini, Joret, Jungers, and Munro~\cite{CFJJM13} later showed that one can take $c=2$, which is best possible. 

\begin{theorem}[Kahn and Kim~\cite{KK95}, Cardinal et al.~\cite{CFJJM13}]
\label{thm:ent}
$$\itlb (P) \leq \lb (P) \leq 2\cdot \itlb (P).$$
\end{theorem}

\paragraph*{Chain and order polytopes.}

We let $P=([n],\leq_P)$.
A point $y\in\mathbb{R}^n$ such that $\forall i,j: i\leq_P j\implies y_i\leq y_j$ is said to be {\em consistent} with $P$.
The {\em order polytope} $\mathcal{O}(P)$ of $P$ is the set of points in $[0,1]^n$ that are consistent with $P$.
With a slight abuse of notation, we also let $\mathcal{O} (\sigma)$ be the order polytope of the total order defined by a permutation $\sigma$. We have $\mathcal {O} (P)=\bigcup_{\sigma\in\Delta (P)} \mathcal{O}(\sigma)$. From this decomposition of the order polytope into $|\Delta (P)|$ simplices, we can deduce that $\mathcal{O}(P)$ has volume $|\Delta (P)|/n!$. From the following development, we will conclude that the chain polytope $\mathcal{C}(P)$ has the same volume.

With a point $y\in\mathbb{R}^n$ consistent with $P$, we can associate a vector $d(y)\in\mathbb{R}^n$ as follows.
If $i$ is a minimum in $P$, then $d_i(y) \coloneqq y_i$.
Otherwise, we define $d_i(y)$ as the minimum of $y_i - y_j$ over all $j\not= i$ such that $j\leq_P i$.
For a point $y\in\mathcal{O}(P)$, it can be seen that $d(y)\in\C (P)$.
This mapping between the order polytope and the chain polytope was defined by Stanley~\cite{S86}, and is referred to as {\em Stanley's transfer map}.

\begin{theorem}[Stanley~\cite{S86}]
\label{thm:transfer}
The transfer map $d$ is a bijective, piecewise linear map between $\mathcal{O} (P)$ and $\C (P)$.
\end{theorem}

\section{Quantum sorting under partial information}
\label{sec:qsort}

We now consider the problem of sorting under partial information in the quantum decision tree model.
We first formulate the lower bound technique used by H\o{}yer et al.~\cite{HNS02} and Yao~\cite{Y04} in the adversarial
framework developed by Ambainis~\cite{A02} and Barnum, Saks, and Szegedy~\cite{BSS03}.
Then we provide a simple formula for the obtained lower bound, involving a variant of the partial order entropy.

\paragraph*{Quantum query lower bound.}
We consider a real symmetric matrix $\Gamma\in\mathbb{R}^{\Delta^2 (P)}$ indexed by pairs of permutations in $\Delta (P)$. 
Furthermore, for each pair of elements $i,j$, we define a new matrix $\Gamma^{ij}\in\mathbb{R}^{\Delta^2(P)}$ such that 
$\Gamma^{ij}_{\sigma\tau} = 0$ if $(\sigma (i) \leq \sigma (j)) = (\tau (i) \leq \tau (j))$,
hence if the result of the comparison between $i$ and $j$ is the same in both permutations. 
Otherwise, $\Gamma^{ij}_{\sigma\tau} = \Gamma_{\sigma\tau}$.

The adversary lower bound is, up to a constant:
$$
\adv (\sort_P) \coloneqq \max_{\Gamma} \frac{\norm{\Gamma}}{\max_{ij} \norm{\Gamma^{ij}}} , 
$$
where the maximization is over all such real symmetric matrices, and $\norm{\cdot}$ denotes the spectral norm.

\paragraph*{An adversary matrix.}

Given a permutation $\sigma$, we denote by $\sigma^{(k,d)}$ the permutation obtained from $\sigma$ by moving the element 
in position $k+d$ down to position $k$. More precisely, if $\tau = \sigma^{(k,d)}$, then
$$
\sigma^{-1} (i) =
\begin{cases}
\tau^{-1} (k) & \text{if } i=k+d, \\
\tau^{-1} (i+1) & \text{if } k\leq i< k+d, \\
\tau^{-1} (i) & \text{otherwise.}
\end{cases}
$$
Then we let $\Gamma_{\sigma\tau} = 1/d$ when $\tau = \sigma^{(k,d)}$. Otherwise, $\Gamma_{\sigma\tau} = 0$.

With a permutation $\sigma\in \Delta (P)$, we associate a point $(\sigma (1), \sigma (2),\ldots ,\sigma (n))\in\mathbb{R}^n$ that is consistent with $P$.
We use the notation $d_i(\sigma)=d_i((\sigma(1),\sigma(2),\ldots ,\sigma(n)))$ and let $H_q = \sum_{i=1}^q \frac{1}{q}$ denote the $q$-th harmonic number. 
\begin{lemma}
\label{lem:num}
$\norm{\Gamma} \geq E_{\sigma \in \Delta (P)} \left[ \sum_{i\in [n]} H_{d_i(\sigma) - 1} \right]$.
\end{lemma}
\begin{proof}
We have $\norm{\Gamma}\geq v^T \Gamma v$ for any unit vector $v$.
Let $v$ be such that $v_{\sigma} = |\Delta (P)|^{-\frac 12}$ for all $\sigma$.
Then 
$$
\norm{\Gamma}\geq \frac{1}{|\Delta (P)|} \sum_{\sigma \in\Delta (P)} \sum_{i\in [n]} \sum_{d\in [d_i(\sigma) - 1]} \frac{1}{d} = E_{\sigma \in \Delta (P)} \left[ \sum_{i\in [n]} H_{d_i(\sigma) - 1} \right].
$$
\end{proof}

The proof of the following upper bound is given in appendix.
\begin{lemma}
\label{lem:denominator}
$\max_{ij} \norm{\Gamma^{ij}} \leq 2\pi$.
\end{lemma}

The following adversary lower bound for $\sort_P$ follows from Lemmas~\ref{lem:num} and~\ref{lem:denominator}.
\begin{equation}
\label{eqn:qlb}
\qlb (P) \coloneqq E_{\sigma \in \Delta (P)} \left[ \sum_{i\in [n]} H_{d_i(\sigma) - 1} \right].
\end{equation}

\begin{lemma}
\label{lem:qlb}
$\adv (\sort_P)\geq c\cdot \qlb(P),$
for some universal constant $c>0$. 
\end{lemma}

\paragraph*{From $\qlb$ to the entropy.}

We first rewrite the harmonic number involved in the above expression of $\qlb$.

\begin{lemma}
\label{lem:exp-ln}
For every poset $P$ on $n$ elements, $\sigma\in\Delta(P)$ and $1\leq i\leq n$, we have
\begin{align*}
 H_{d_i(\sigma)-1}=H_n+\Exp{y\in \mathcal{O}(\sigma)}{\ln d_i(y)}
\end{align*}
\end{lemma}
\begin{proof}
We prove that $d_i(y)$ has a probability density $f_{y\in \mathcal{O} (\sigma )}[d_i(y)=s]=f_{n,d_i(\sigma)-1}[s]$,
where 
$$
f_{n,k}[s]=n\binom{n-1}{k} s^k(1-s)^{n-k-1}.
$$
We then have 
$$\Exp{y\in \mathcal{O} (\sigma )}{\ln d_i(y)}= H_{d_i(\sigma)-1}-H_n.$$
The details are given in Lemma~\ref{lem:useful-integrals} and Lemma~\ref{lem:density-order-polytope} in appendix.
\end{proof}

We now give an {\it exact} rewriting of the quantum lower bound $\qlb$.
Let
$$
\QH (P) = \Exp{z \in \C(P)}{-\frac 1n \sum_{i\in [n]} \ln z_i}.
$$
The reader is encouraged to compare this expression with the one for the entropy in Equation~\eqref{eqn:entropy}.
The following result provides the quantum analogue of Equation~\eqref{eqn:lb}.

\begin{theorem}
  \label{thm:improved}
  $\qlb(P) = n(H_n - \QH (P))$.
\end{theorem}
\begin{proof}
\begin{eqnarray*}
\qlb (P) & = & \Exp{\sigma\in\Delta(P)}{\sum_{i\in [n]} H_{d_i(\sigma)-1}} \\
& = & \Exp{\sigma\in\Delta(P)}{\sum_{i\in [n]} \left( H_n+\Exp{y\in \mathcal{O} (\sigma )}{\ln d_i(y)} \right)} \text{\ (from\ Lemma~\ref{lem:exp-ln})}\\
& = & n \left( H_n+\Exp{y\in\mathcal{O}(P)}{\frac 1n \sum_{i\in [n]} \ln d_i(y) }\right) \\
& = & n \left( H_n +\Exp{z\in\C (P)}{ \frac 1n \sum_{i\in [n]} \ln z_i }\right) \text{\ (from\ Theorem~\ref{thm:transfer})}\\
& = & n (H_n - \QH (P)).
\end{eqnarray*}
\end{proof}

\paragraph*{Discussion.}
The quantity $\QH (P)$ is an averaged version of the entropy $\H(P)$ which can be shown to lie in the interval $[1,H_n]$.
Yao~\cite{Y04} already established a relation between the adversary lower bound and $\QH (P)$.
Precisely, he proved that $$\qlb (P)\geq \Omega (n(\ln n - \QH (P))).$$
This inequality is not sufficient to get rid of the linear term $-c'n$ in Theorem~\ref{thm:yao}.   
Theorem~\ref{thm:improved} strengthens the relation to an equality when $\ln n$ is replaced by $H_n$ in the right-hand side.
This exact reformulation allows us to analyze the bound on a wide class of posets.

\section{A proof of Yao's conjecture for series-parallel posets}
\label{sec:sp}

Given  two posets $P$ and $Q$ with disjoint element sets, the {\em series composition}, or {\em ordinal sum} $P\oplus Q$ is the poset 
on the union of the element sets of $P$ and $Q$ such that $x \leq_{P\oplus Q} y$ if and only if one of the following holds:
\begin{enumerate}
\item $x \leq_P y$,
\item $x \leq_Q y$,
\item $x$ belongs to $P$ and $y$ to $Q$.
\end{enumerate} 
Given  two posets $P$ and $Q$ with disjoint element sets, the 
{\em parallel composition}, or {\em direct sum} $P + Q$ is the poset 
on the union of the element sets of $P$ and $Q$ such that $x \leq_{P+Q} y$ if and only if 
either $x \leq_P y$ or $x \leq_Q y$.
A poset is a {\em series-parallel} poset if and only if it is a singleton, or it can be obtained
by a series or parallel composition of two series-parallel posets.
The definition is illustrated in Figure~\ref{fig:sp}.

\begin{figure}
\begin{center}
\begin{tikzpicture}[scale=.7,auto,swap]
\begin{scope}[xshift=0cm]
\foreach \left/\right in {{(1,7)/(7,9)}, {(1,4)/(7,6)}, {(1,1)/(7,3)}} 
   \filldraw[fill=red!40, fill opacity = .3, rounded corners] \left rectangle \right;
\foreach \left/\right in {{(2.5,7.5)/(3.5,8.5)}, {(4.5,7.1)/(5.5,8.9)}, {(1.5,4.5)/(2.5,5.5)}, {(3.5,4.5)/(4.5,5.5)}, {(5.5,4.5)/(6.5,5.5)}}
   \filldraw[fill=blue!70, fill opacity = .3, rounded corners] \left rectangle \right;
\foreach \pos/\name/\outdeg in {{(3,8)/a}, {(5,8.5)/b}, {(5,7.5)/c}, {(2,5)/d}, {(4,5)/e}, {(6,5)/f}, {(4,2)/g}}
   \node[vertex] (\name) at \pos {};
\foreach \source/\dest in {a/d,a/e,a/f,b/c,c/d,c/e,c/f,d/g,e/g,f/g}
   \path[edge] (\source) -- (\dest);
\end{scope}
\end{tikzpicture}
\end{center}
\caption{\label{fig:sp}A series-parallel poset of the form $\circ \oplus (\circ +\circ +\circ) \oplus (\circ + (\circ\oplus\circ))$, where $\circ$ is a singleton.}
\end{figure}
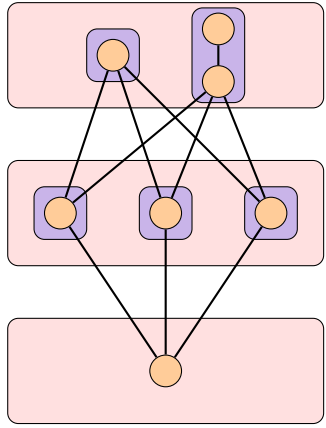

We consider the behavior of $\qlb$ and $\QH$ under the two composition operations.

\begin{lemma}[Quantum lower bound \& series compositions]
\label{lem:lbser}
Let $P$ and $Q$ be two disjoint posets. 
Then 
$$
\qlb (P\oplus Q) = \qlb (P) + \qlb(Q).
$$
\end{lemma} 
\begin{proof}
Let $P$ be defined on $\{1, \dots, n_1\}$, $Q$ on $\{n_1 + 1, \dots, n_1 + n_2\}$, and set $n:= n_1+n_2$.  
We use the original formulation of the quantum lower bound given in~\eqref{eqn:qlb}, and the fact that any linear extension
of $P\oplus Q$ consists of a linear extension of $P$ followed by a linear extension of $Q$. Furthermore, the value of $d_i(\sigma)$
for an element $i\in P$ and a linear extension $\sigma\in\Delta(P\oplus Q)$ is the same as the one for the corresponding $\sigma$ in $\Delta(P)$.
Therefore,
\begin{eqnarray*}
\qlb (P\oplus Q) & = & \Exp{\sigma\in\Delta(P\oplus Q)}{\sum_{i\in [n]} H_{d_i(\sigma)-1}} \\
& = & \Exp{\sigma\in\Delta(P\oplus Q)}{\sum_{i\in [n_1]} H_{d_i(\sigma)-1} + \sum_{i=n_1+1}^n H_{d_i(\sigma)-1}} \\
& = & \Exp{\sigma\in\Delta(P)}{\sum_{i\in [n_1]} H_{d_i(\sigma)-1}} + \Exp{\sigma\in\Delta(Q)}{\sum_{i=n_1+1}^n H_{d_i(\sigma)-1}} \\
& = & \qlb (P) + \qlb (Q) .
\end{eqnarray*}
\end{proof}

Analyzing parallel compositions using the quantum lower bound as formulated in~\eqref{eqn:qlb} seems difficult. 
However, our reformulation as a function of $\QH$ makes this case particularly easy.

\begin{lemma}[$\QH$ and parallel composition]
  \label{lem:qhpar}
  Let $P$ and $Q$ be two posets with element sets $\{1, \dots, n_1\}$ and $\{n_1 + 1, \dots, n_1 + n_2\}$ respectively, and let $n:= n_1+n_2$.  
  Then
  $$
  \QH (P + Q) = \frac{n_1}n \QH(P) + \frac{n_2}n \QH(Q) .
  $$
\end{lemma}
\begin{proof}
  A chain in the poset $P+Q$ is always fully contained in either $P$ or $Q$. Therefore,
  $z\in \C(P+Q)$ if and only if $(z_1,\ldots ,z_{n_1})\in \C(P)$ and $(z_{n_1+1},\ldots ,z_n)\in\C(Q)$.
  Hence,
  \begin{eqnarray*}
    \QH (P+Q) & = & \Exp{z \in \C(P+Q)}{-\frac 1n \sum_{i\in [n]} \ln z_i} \\
    & = & \Exp{z \in \C(P+Q)}{-\frac 1n \sum_{i\in [n_1]}  \ln z_i} + \Exp{z \in \C(P+Q)}{-\frac 1n \sum_{i=n_1+1}^n \ln z_i }\\
    & = & \Exp{z \in \C(P)}{-\frac 1n \sum_{i\in [n_1]}  \ln z_i} + \Exp{z \in \C(Q)}{-\frac 1n \sum_{i=n_1+1}^n \ln z_i }\\
    & = & \frac{n_1}n \QH(P) + \frac{n_2}n \QH(Q) .
    \end{eqnarray*}
  \end{proof}

\begin{lemma}[Quantum lower bound \& parallel compositions]
\label{lem:lbpar}
Let $P$ and $Q$ be two posets with element sets $\{1, \dots, n_1\}$ and $\{n_1 + 1, \dots, n_1 + n_2\}$ respectively, and let $n:= n_1+n_2$.  
Then 
$$
\qlb (P + Q) = \qlb (P) + \qlb(Q) + nH_n - n_1 H_{n_1} - n_2H_{n_2}.
$$
\end{lemma} 
\begin{proof}
  Here we consider the reformulation of $\qlb$ given in Theorem~\ref{thm:improved}.
\begin{eqnarray*}
  \qlb (P+Q) & = & n(H_n - \QH (P+Q)) \\
  & = & n\left(H_n - \left(\frac{n_1}n \QH (P) + \frac{n_2}n \QH(Q)\right)\right) \text{(from\ Lemma~\ref{lem:qhpar})}\\
  & = & nH_n - n_1\QH(P) - n_2\QH(Q) \\
  & = & \qlb (P) + \qlb(Q) + nH_n - n_1 H_{n_1} - n_2H_{n_2}.
\end{eqnarray*}
\end{proof}

Before proving our main result, we need the following technical lemma, obtained from Stirling's formula.

\begin{lemma}
  \label{lem:tech}
  There exists a constant $c>0$ such that for all integers $n_1,n_2 \geq 1$,
  $$
  (n_1+n_2)H_{n_1+n_2} - n_1H_{n_1} - n_2H_{n_2}\geq c \ln {n_1+n_2\choose n_1}.
  $$
  \end{lemma}

\begin{lemma}
\label{lem:main}
  For a series-parallel poset $S$, we have
  $\qlb (S) \geq c\cdot \itlb (S)$ for some universal constant $c>0$.
  \end{lemma}
\begin{proof}
  We prove the result with the constant $c$ from Lemma~\ref{lem:tech}, by induction on the number $n$ of elements of $S$.
  For $n=1$, we can check that both bounds are equal to zero.
  For $n>1$, we know that $S$ is obtained via a series or a parallel composition
  of two posets $P$ and $Q$ with respective element sets $[n_1]$ and $\{n_1+1,\ldots ,n\}$.
  We let $n_2=n-n_1$ and suppose that the statement holds for $P$ and $Q$.

  Let us first consider the case of the series composition, where $S=P\oplus Q$. We have
  $$
  |\Delta (P\oplus Q)| = |\Delta (P)|\cdot |\Delta (Q)|.
  $$
  We can apply Lemma~\ref{lem:lbser} and the induction hypothesis:
  \begin{eqnarray*}
    \qlb (P\oplus Q) & = & \qlb (P) + \qlb (Q) \\
    & \geq & c\cdot \itlb (P) + c\cdot \itlb (Q) \\
    & = & c\cdot \ln (|\Delta (P)|\cdot |\Delta (Q)|) \\
    & = & c\cdot \ln (|\Delta (P\oplus Q)|)  = c\cdot\itlb (S).
  \end{eqnarray*}

  For the case where $S=P+Q$, we have
  $$
    |\Delta (P+Q)| = |\Delta (P)|\cdot |\Delta (Q)|\cdot {n_1+n_2\choose n_1}.
  $$
  We can apply Lemma~\ref{lem:lbpar} and the induction hypothesis:
  \begin{eqnarray*}
    \qlb (P+Q) & = & \qlb (P) + \qlb (Q) + nH_n - n_1H_{n_1} - n_2H_{n_2} \\
    & \geq & c\cdot\itlb (P) + c\cdot\itlb (Q) + nH_n - n_1H_{n_1} - n_2H_{n_2} \\
    & \geq & c\cdot\ln \left(|\Delta (P)|\cdot |\Delta (Q)|\cdot {n_1+n_2\choose n_1} \right) \text{(from\ Lemma~\ref{lem:tech})}\\
    & = & c\cdot\ln (|\Delta (P+Q)|) = c\cdot\itlb (S).
    \end{eqnarray*}\end{proof}

Combining Lemmas~\ref{lem:qlb} and~\ref{lem:main} yields Theorem~\ref{thm:main}.

\paragraph*{Extending the result to a wider class of posets.}

The {\em $N$ poset} on four elements $a,b,c,d$ is such that $a\leq b$, $c\leq b$, $c\leq d$, and all the other pairs are incomparable. Its name comes from the shape of its Hasse diagram.
It is well-known that series-parallel posets are exactly the posets that forbid the $N$ poset as induced subposet~\cite{VTL82}.

We now show that the inequality in Lemma~\ref{lem:main} also holds for posets that are far from series-parallel, in the sense that they have $\Omega (n^4)$ induced $N$ subposets.
Let $N_k,k\in\mathbb{N}$ be the poset on $n=4k$ elements obtained by replacing each of the four elements of an $N$ poset by a chain of length $k$. We denote these chains by $A=(a_i)_{i\in [k]}, B=(b_i)_{i\in [k]}, C=(c_i)_{i\in [k]}, D=(d_i)_{i\in [k]}$, respectively. For any quadruple $i,j,k,\ell \in [k]^4$, the poset induced by $a_i,b_j,c_k,d_{\ell}$ is an $N$ poset, hence $N_k$ has $k^4=\Omega (n^4)$ induced $N$ subposets.

\begin{lemma}
$\itlb (N_k)=\Theta (k)$.
\end{lemma}
\begin{proof}
  We have
  $$|\Delta (N_k)| > |\Delta (A+B)| = {2k\choose k} \sim 4^k/\sqrt{\pi k}.$$
  Also,
  $$|\Delta (N_k)| < |\Delta ((A\oplus B) + (C\oplus D))| = {4k\choose 2k} \sim 4^{2k}/\sqrt{\pi 2 k}.$$
  
  Hence $\itlb (N_k) = \ln (|\Delta (N_k)|) = \Theta (k)$.
\end{proof}

\begin{lemma}
$\qlb (N_k)=\Omega (k)$.
\end{lemma}
\begin{proof}
  First note that if a poset $Q$ extends a poset $P$, then $\qlb (P)\geq \qlb (Q)$.
  We notice that $(A+B)\oplus(C+D)$ extends $N_k$.
Therefore,
    \begin{eqnarray*}
    \qlb (N_k) & \geq & \qlb ((A+B)\oplus(C+D)) \\
    & = & \qlb (A+B) + \qlb (C+D) \text{\ (from\ Lemma~\ref{lem:lbser})}\\
    & = & 2\cdot (2kH_{2k} - 2kH_k) \text{\ (from\ Lemma~\ref{lem:lbpar})}\\
    & \geq & c\cdot \ln {2k\choose k} \text{\ (from\ Lemma~\ref{lem:tech})}\\
    & = & \Omega (k).
    \end{eqnarray*}
\end{proof}

A poset is an {\em extended series-parallel} poset if and only if it is either (i) a singleton, (ii) isomorphic to $N_k$ for some $k$, or (iii) it can be obtained by a series or parallel composition of two extended series-parallel posets.
The two lemmas directly imply the following analogue of Theorem~\ref{thm:main} to extended series-parallel posets.
\begin{corollary}
The quantum query complexity of sorting under partial information, given an extended series-parallel poset $P$, is at least $c\log (|\Delta (P)|)$ for some universal constant $c>0$.
\end{corollary}

\section*{Conclusion}

The analyses of $\sort_P$ in the classical and quantum cases rely on the following quantities, where $h(z)=-\frac 1n \sum_{i\in [n]} \ln z_i$.
\begin{center}
\begin{tabular}{|c|c|}
\hline 
 Classical & Quantum \\
\hline 
 $\H(P) = \min_{z \in \C(P)} h(z)$ & $\QH (P) = E_{z \in \C(P)} \left[ h(z)\right]$ \\
 $\lb (P)= n (\ln n - \H(P))$ & $\qlb (P)=n (H_n - \QH (P))$ \\
\hline 
\end{tabular}
\end{center}

Our findings support the conjecture that the two lower bounds $\lb(P)$ and $\qlb(P)$ are within a constant factor of each other for all posets $P$.
Proving this would require to better understand how the quantities $\H(P)$ and $\QH(P)$ behave relative to each other.
However, we seem to be lacking tools to analyze the quantity $E_{z \in C}\left[ h(z)\right]$ defined on an arbitrary convex corner $C$.
In particular, unlike the entropy, it is not monotone with respect to inclusion of $C$.

\section*{Acknowledgements} 
The authors thank Samuel Fiorini and Davy Paindaveine for helpful discussions. 

\renewcommand{\em}{\it}
\bibliographystyle{plain}
\bibliography{qsortref}

\appendix

\section{Proof of Lemma~\ref{lem:denominator}}

\begin{proof}[Proof of Lemma~\ref{lem:denominator}] 
Let $A$ be the Hilbert matrix with $A_{kl}=\frac{1}{k+l-1}$ for $1\leq k,l<n-1$.   
Let $B$ be the matrix with entries $B_{kl}=\frac{\delta[k+l\leq n]}{k+l-1}$ for $1\leq k,l<n-1$. 
Note that $\norm{B}\leq\norm{A}\leq \pi$. 
We will show that $\norm{\Gamma^{jj'}} \leq 2\norm{B}$ for all $j, j'$, which implies the lemma. 

We first note that for $\tau=\sigma^{(k,d)}$, and assuming $\sigma(j)>\sigma(j')$ the matrix element 
$\Gamma^{jj'}_{\sigma\tau}$
is non-zero only if
\begin{align*}
\left\{
\begin{array}{l}
j=\sigma^{-1}(k+d)=\tau^{-1}(k)\\
j'=\sigma^{-1}(k+i)=\tau^{-1}(k+i+1)
\end{array}
\right.
\end{align*}
for some $0\leq i\leq d-1$. Indeed, in that case we have $\sigma(j)>\sigma(j')$ and $\tau(j)<\tau(j')$. Therefore, 
\begin{align*}
 \Gamma^{jj'}
=\sum_\sigma\sum_{k=1}^{n-1}\sum_{d=1}^{n-k}\frac{1}{d}\left[\sum_{i=0}^{d-1}\delta[\sigma(j)=k+d]\cdot\delta[\sigma(j')=k+i]\right]\left[\ketbra{\sigma}{\sigma^{(k,d)}}+\ketbra{\sigma^{(k,d)}}{\sigma}\right],
\end{align*}
where the last term covers cases where $\sigma(j)<\sigma(j')$. Manipulating this expression, we obtain
\begin{align*}
  \Gamma^{jj'}
  &= \sum_{d=1}^{n-1}\frac{1}{d}\sum_{i=0}^{d-1}\sum_{\sigma:\sigma(j')=\sigma(j)-d+i}\left[\ketbra{\sigma}{\sigma^{(\sigma(j)-d,d)}}+\ketbra{\sigma^{(\sigma(j)-d,d)}}{\sigma}\right]\\
 &=\sum_{l,m=1}^{n-1}\frac{\delta[l+m\leq n]}{l+m-1}\sum_{\sigma:\sigma(j')=\sigma(j)-l}\left[\ketbra{\sigma}{\sigma^{(\sigma(j)-l-m+1,l+m-1)}} + \ketbra{\sigma^{(\sigma(j)-l-m+1,l+m-1)}}{\sigma}\right]\\
&=\sum_{l,m=1}^{n-1}B_{ml}\sum_{\sigma:\sigma(j')=\sigma(j)-l}\left[\ketbra{\sigma}{\sigma_{(l,m)}}+\ketbra{\sigma_{(l,m)}}{\sigma}\right],
\end{align*}
where the second equality follows from the change of variables $l=d-i$ and $m=i+1$, and in the last line we have used the notation $\sigma_{(l,m)}=\sigma^{(\sigma(j)-l-m+1,l+m-1)}$. Note that when $\sigma$ runs over all permutations such that $\sigma(j')=\sigma(j)-l$, then $\tau=\sigma_{(l,m)}$ runs over all permutations such that  $\tau(j')=\tau(j)+m$.

By definition, we have $\norm{\Gamma^{jj'}}=\max_{\ket{v}}|\bra{v}\Gamma^{jj'}\ket{v}|$, where the maximization is over all unit vectors $\ket{v}$. For any such vector, let $v_\sigma=\braket{\sigma}{v}$, and $\alpha_l$ and $\beta_m$ be defined as
\begin{align*}
 \alpha_l&=\sqrt{\sum_{\sigma:\sigma(j')=\sigma(j)-l}v_\sigma^2} &
 \beta_m&=\sqrt{\sum_{\tau:\tau(j')=\tau(j)+m}v_\tau^2}
\end{align*}
Then, we have
\begin{align*}
 \abs{\bra{v}\Gamma^{jj'}\ket{v}}
\leq 2\sum_{l,m=1}^{n-1}B_{ml}\sum_{\sigma:\sigma(j')=\sigma(j)-l}\abs{v_{\sigma}v_{\sigma_{(l,m)}}}
\leq 2\sum_{l,m=1}^{n-1}B_{ml}\alpha_l\beta_m
\leq 2\norm{B},
\end{align*}
where in the first inequality we have used the fact that $B_{ml}\geq 0$, and the second inequality follows from Cauchy-Schwarz.
\end{proof}

\section{Order statistics}

For $0\leq k<n$, we define the probability density $f_{n,k}$ for random variable $z\in[0,1]$ as
\begin{align*}
 f_{n,k}[z=s]=n\binom{n-1}{k} s^k(1-s)^{n-k-1}.
\end{align*}
Note that $f_{1,0}$ is the density of a uniformly distributed random variable over $[0,1]$.

The following integrals will be useful. 

\begin{lemma}
\label{lem:useful-integrals}
\begin{align}
  I_{n,k}(s)=k\binom{n}{k}\int_{0}^{1-s}dt\ t^{n-k}(1-t-s)^{k-1}&=(1-s)^n&\forall 0\leq s\leq 1\label{eq:Ink}\\
 J_{n,k}(s)=\Prob{z\sim f_{n,k}}{z\leq 1-s}=n\binom{n-1}{k} \int_0^{1-s}dt\ t^k(1-t)^{n-k-1}&=\sum_{l=k+1}^n\binom{n}{l}s^{n-l}(1-s)^l&\forall 0\leq s\leq 1\label{eq:Jnk}\\
H_{n,k}=\Exp{z\sim f_{n,k}}{\ln z}=n\binom{n-1}{k} \int_0^1 dt\ t^k(1-t)^{n-k-1}\ln t&=H_k-H_n\label{eq:Hnk}
 \end{align}
\end{lemma}

\begin{lemma}
\label{lem:density-order-polytope}
If $z\in\mathbb{R}^n$ is uniformly distributed over the simplex $0\leq z_1\leq z_2\leq\ldots\leq z_n\leq 1$, then
\begin{align*}
 f[z_{i+d}-z_i=s]=f_{n,d-1}[s].
\end{align*}
\end{lemma}

\section{Proof of Lemma~\ref{lem:useful-integrals}}

\begin{proof}
 For Equation~(\ref{eq:Ink}), let us first evaluate $I_{n,1}$
\begin{align*}
 I_{n,1}&=n\int_{0}^{1-s}dt\ t^{n-1}=\left[t^n\right]_{0}^{1-s}=(1-s)^n
\end{align*}
For $k>1$, integration by parts leads to
\begin{align*}
 I_{n,k}&=\frac{n!}{(k-1)!(n-k)!}\int_{0}^{1-s}dt\ t^{n-k}(1-t-s)^{k-1}\\
&=\frac{n!}{(k-1)!(n-k+1)!}\left[t^{n-k+1}(1-t-s)^{k-1}\right]_0^{1-s}
+\frac{n!}{(k-2)!(n-k+1)!}\int_{0}^{1-s}dt\ t^{n-k+1}(1-t-s)^{k-2}\\
&=0+(k-1)\binom{n}{k-1}\int_{0}^{1-s}dt\ \frac{t^{n-k+1}}{(n-k+1)!}(1-t-s)^{k-2}\\
&=I_{n,k-1}
\end{align*}
By induction, we therefore have $I_{n,k}=I_{n,1}=(1-s)^n$.
 For Equation~(\ref{eq:Jnk}), we first evaluate $J_{n,0}$
\begin{align*}
 J_{n,0}&=n\int_0^{1-s}dt\ (1-t)^{n-1}=-\left[(1-t)^n\right]_0^{1-s}=1-s^n
\end{align*}
For $k>0$, integration by parts leads to
\begin{align*}
J_{n,k}&=n\binom{n-1}{k} \int_0^{1-s}dt\ t^k(1-t)^{n-k-1}\\ 
&=(n-k)\binom{n}{k}\int_0^{1-s}dt\ t^k(1-t)^{n-k-1}\\
&=\binom{n}{k}\left\{-\left[t^k(1-t)^{n-k}\right]_0^{1-s}+k\int_0^{1-s}dt\ t^{k-1}(1-t)^{n-k}\right\}\\
&=-\binom{n}{k}s^{n-k}(1-s)^k+n\binom{n-1}{k-1}\int_0^{1-s}dt\ t^{k-1}(1-t)^{n-k}\\
&=-\binom{n}{k}s^{n-k}(1-s)^k+J_{n,k-1}
\end{align*}
By induction, this leads to
\begin{align*}
 J_{n,k}&=-\sum_{l=1}^k\binom{n}{l}s^{n-l}(1-s)^l+J_{n,0}\\
&=1-\sum_{l=0}^k\binom{n}{l}s^{n-l}(1-s)^l\\
&=\sum_{l=k+1}^n\binom{n}{l}s^{n-l}(1-s)^l
\end{align*}
where in the last line we have used the fact that
\begin{align*}
 \sum_{l=0}^n\binom{n}{l}s^{n-l}(1-s)^l=\left[s+(1-s)\right]^n=1
\end{align*}
 For Equation~(\ref{eq:Hnk}), we first evaluate $H_{k+1,k}$.  
\begin{align*}
 H_{k+1,k}&=(k+1) \int_0^1 dt\ t^k\ln t
=\left[t^{k+1}\ln t\right]_0^1-\int_0^1 dt\ t^{k+1}\frac{1}{t}\\
&=-\frac{1}{k+1}\left[t^{k+1}\right]_0^1=-\frac{1}{k+1}=H_k-H_{k+1}
\end{align*}
For $n>k+1$, we need to evaluate
\begin{align*}
 H_{n,k}&=(n-k)\binom{n}{k} \int_0^1 dt\ t^k(1-t)^{n-k-1}\ln t,
\end{align*}
which we integrate by parts by setting $(n-k)t^k(1-t)^{n-k-1}\ln t=u\cdot v'$ with
\begin{align*}
u&=-t^{n+1}\ln t & u'&=-t^n\left[1+(n+1)\ln t\right]\\
v&=\left(\frac{1-t}{t}\right)^{n-k} & v'&=-(n-k)\frac{(1-t)^{n-k-1}}{t^{n-k+1}}
\end{align*}
which leads to
\begin{align*}
 H_{n,k}&=-\binom{n}{k}\left[t^{k+1}(1-t)^{n-k}\ln t\right]_0^1+\binom{n}{k}\int_0^1 dt\ t^k(1-t)^{n-k}\left[1+(n+1)\ln t\right]\\
&=0+\binom{n}{k}\int_0^1 dt\ t^k(1-t)^{n-k} + \binom{n}{k}\int_0^1 dt\ (n+1)t^k(1-t)^{n-k}\ln t\\
&=\frac{n+1}{n+1}\binom{n}{k}\int_0^1 dt\ t^k(1-t)^{n-k}+(n+1-k)\binom{n+1}{k}\int_0^1 dt\ t^k(1-t)^{n-k}\ln t\\
&=\frac{1}{n+1}J_{n+1,k}(0)+H_{n+1,k}\\
&=\frac{1}{n+1}+H_{n+1,k}
\end{align*}
This implies
\begin{align*}
 H_{n+1,k}&=H_{n,k}-\frac{1}{n+1},
\end{align*}
which by induction leads to
\begin{align*}
 H_{n,k}&=H_{k+1,k}-\sum_{m=k+2}^{n}\frac{1}{m}\\
&=H_k-H_{k+1}-(H_n-H_{k+1})\\
&=H_k-H_n
\end{align*}
\end{proof}

\section{Proof of Lemma~\ref{lem:density-order-polytope}}

\begin{proof}
 The joint probability density of $z=(z_1,z_2,\ldots,z_n)$ is given by
\begin{align*}
 f[z_1=s_1,\ldots,z_n=s_n]=n!\ \delta(s_1\leq s_2\leq \ldots\leq s_n).
\end{align*}
By successive integrations, we obtain for $i<j$
\begin{align*}
 f[z_i=s_i,z_j=s_j]=\frac{n!}{(i-1)!(n-j)!(j-i-1)!}s_i^{i-1}(1-s_j)^{n-j}(s_j-s_i)^{j-i-1}\delta(s_i<s_j).
\end{align*}
Therefore, for $j=i+d$, we have
\begin{align*}
& f[z_{i+d}-z_i=s] \\
&= \int_0^1 ds_i\int_0^1 ds_{i+d} f[z_i=s_i,z_j=s_j]\ \delta[s_{i+d}-s_i=s]\\
&=\frac{n!}{(i-1)!(n-i-d)!(d-1)!}\\
& \int_0^1 ds_i\int_0^1 ds_{i+d}\ s_i^{i-1}(1-s_{i+d})^{n-i-d}(s_{i+d}-s_i)^{d-1}\ \delta[s_{i+d}-s_i=s]\\
&=\frac{n!}{(i-1)!(n-i-d)!(d-1)!}s^{d-1}\int_0^{1-s} ds_i\ s_i^{i-1}(1-s_i-s)^{n-i-d}
\end{align*}
We recognize integral $I_{n',k}(s)$ from Lemma~\ref{lem:useful-integrals} with $n'=n-d$ and $k=n-i-d+1$, so that
\begin{align*}
  f[z_{i+d}-z_i=s]&=\frac{n!}{(i-1)!(n-i-d)!(d-1)!}s^{d-1}\frac{(1-s)^{n-d}}{(n-i-d+1)\binom{n-d}{n-i-d+1}}\\
&=n\binom{n-1}{d-1}s^{d-1}(1-s)^{n-d},
\end{align*}
which coincides with the definition of $f_{n,d-1}[s]$.
\end{proof}

\end{document}